\documentclass[letterpaper, 10pt, conference,dvipsname]{ieeeconf}
\usepackage{cite}
\usepackage[dvipsnames]{xcolor}
\IEEEoverridecommandlockouts                              
\let\labelindent\relax
\usepackage{enumitem}
\usepackage{psfrag,amsthm,amsmath,amsfonts,verbatim, mathtools,enumerate,algorithm,dsfont,todonotes}
\usepackage{cite}
\usepackage{diagbox}
\usepackage{times,algpseudocode,bm,tikz}
\usepackage{booktabs}
\allowdisplaybreaks[2]
\newcommand{\ones}{\mathbf 1}
\newcommand{\reals}{{\mathbb{R}}}

\newcommand{\rank}{\mathop{\bf rank}}

\newcommand{\norm}[1]{\left\lVert#1\right\rVert}
\newcommand{\mnorm}[1]{{\left\vert\kern-0.25ex\left\vert\kern-0.25ex\left\vert #1 
    \right\vert\kern-0.25ex\right\vert\kern-0.25ex\right\vert}}

\newcommand{\mc}{\mathcal}

\newcommand{\G}{\mathcal{G}}

\newcommand{\E}{\mathcal{E}}

\newtheorem{definition}{Definition} 
\newtheorem{theorem}{Theorem}
\newtheorem{lemma}{Lemma}
\newtheorem{corollary}{Corollary}

\newtheorem{proposition}{Proposition}
\newtheorem{assumption}{Assumption}

\algnewcommand\algorithmicforeach{\textbf{for each}}
\algdef{S}[FOR]{ForEach}[1]{\algorithmicforeach\ #1\ \algorithmicdo}

\begin{document}
\title{\Large \bf Sensitivity Analysis for Markov Decision Process Congestion Games}
\author{Sarah H. Q. Li$^{1}$, Daniel Calderone$^{2}$, Lillian Ratliff$^{2}$, Beh\c cet\ A\c c\i kme\c se$^{1}$
\thanks{*This work was is supported by NSF award CNS-1736582.}
\thanks{$^{1}$Authors are with the William E. Boeing  Department of Aeronautics and Astronautics, University of Washington, Seattle. 
        {\tt\small sarahli@uw.edu}
        {\tt\small behcet@uw.edu}}%
\thanks{$^{2}$Authors are with the Department of Electrical Engineering, University of Washington, Seattle.
        {\tt\small djcal@uw.edu }
        {\tt\small ratliffl@uw.edu }}%
}
\maketitle
\begin{abstract}
We consider a non-atomic congestion game where each decision maker performs selfish optimization over states of a common MDP. The decision makers optimize for their own expected cost, and influence each other through congestion effects on the state-action costs. 
We analyze the sensitivity of MDP congestion game equilibria to uncertainty and perturbations in the state-action costs by applying an implicit function type analysis. The occurrence of a stochastic Braess paradox is defined and analyzed based on sensitivity of game equilibria and demonstrated in simulation. We further analyze how the introduction of stochastic dynamics affects the magnitude of Braess paradox in comparison to  deterministic dynamics. 
\end{abstract}

\section{Introduction}\label{sec:introduction}
\emph{Markov decision process (MDP) congestion games} have been successfully used to model distributions of selfish decision makers when competing for finite resources~\cite{calderone2017markov}. In particular, MDP congestion games introduce stochastic dynamics in congestion games by mapping user inputs to probabilistic outcomes. An equilibrium concept similar to Wardrop equilibrium of routing games \cite{calderone2017infinite}, \emph{MDP Wardrop equilibrium} describes steady-state population behaviour at which no players can optimize their expected state-action costs through further changes in their decision strategies. 

In modelling a physical process as a game, the game equilibrium \emph{approximates} the true steady-state of the physical process; this is because models inherently cannot predict the physical process to full accuracy. The underlying assumption is that the modelling errors cause negligible deviations of prediction from physical equilibrium. However, this is false if the steady-state distribution is \emph{sensitive} to changes in the modelling parameters. This motivates our study of sensitivity of MDP congestion game to state-action costs.    

In this paper, we quantify sensitivity for the occurrence of \emph{stochastic Braess paradox}, and relate the paradox to its deterministic counterpart. We also define and derive conditions for MDP dynamics and state-action costs under which our sensitivity analysis is valid. Finally 
we bound the sensitivity of a stochastic MDP congestion game in terms of the sensitivity of its deterministic counterpart. 

Here we'd also like to emphasize why we consider the sensitivity of Wardrop equilibrium to the state-action cost parameters. 
In utilizing MDP congestion game models to forecast steady-state behaviour of a physical system, state-action costs are often parameterized by experimental data, which has uncertainty. When this uncertainty is bounded, it is natural to consider bounding the resulting deviation of true equilibrium from the predicted equilibrium. Secondly, the sensitivity of game equilibrium is highly relevant to Stackelberg games for the  leader, who may utilize the sensitivity information to derive an optimal action sequence for its own objective~\cite{shiau2009optimal}. Finally, when a game designer with a certain `budget' for changing the cost function attempts to alter an existing game equilibrium to maximize an external objective, it's important to know the optimal change with respect to designer's alternative objective.

We review existing literature on sensitivity and MDP congestion games in section~\ref{sec:related work}. In section~\ref{sec:preliminaries}, MDP congestion game and related concepts are defined. Sensitivity results and stochastic Braess paradox characterizations are given in section~\ref{sec:sensitivityAnalysis}. We analyze stochasticity's effect on paradox sensitivity in section~\ref{sec:stochasticity}. Finally, simulations demonstrating stochastic Braess paradox and the sensitivity analysis is shown in section~\ref{sec:wheatstoneExamples}. 
\section{Related work}\label{sec:related work}
MDP congestion games \cite{calderone2017markov,calderone2017infinite} combine features of non-atomic routing games  \cite{wardrop1952,beckmann1952continuous,patriksson2015traffic}, i.e.~where decision makers influence each other's edge costs through congestion effects over a network---and \emph{stochastic games}~\cite{shapley1953stochastic,mertens1981stochastic}---i.e.~where each decision maker solves an MDP. 

Our analysis resembles sensitivity work on Wardrop equilibria in traffic assignment literature~\cite{Tobin1988,qiu1992sensitivity,Patriksson2004}, where extensive research exist on both the sensitivity of Wardrop equilibria and a related problem of network design with respect to optimal user equilibrium~\cite{yamada2015freight, bar2013computational}. Efficiency of Wardrop equilibria leads to a paradoxical phenomenon known as Braess paradox~\cite{braess1968paradoxon}, whose occurrence is linked to the underlying network of system dynamics~\cite{milchtaich2006network}. 

To incorporate randomness in the traffic assignment model, a variety of probabilistic models were analyzed. Approximation algorithms have been derived for networks where uncertainty exists in user demand~\cite{ukkusuri2007robust}, in user dynamics as logit model~\cite{liu2015global}, and in perceived cost function as normal error distribution~\cite{clark2002sensitivity}. Sensitivity of other network games to modelling parameters have also been studied in~\cite{parise2019variational}. Our work is fundamentally different from previous work due to our assumption: we consider exclusively on uncertain dynamics, and instead of modelling uncertainty with explicit probability distributions, we describe dynamics with MDPs, which can be interpreted as a discretization of an arbitrary probability distribution. The addition of MDP dynamics then requires additional treatment as described in later sections. 
\section{Preliminaries}\label{sec:preliminaries}
We introduce MDP congestion game framework from an individual decision maker's perspective and define a variational inequality-style game equilibria. From a system-level perspective, MDP congestion game is formulated as a potential game with a hypergraph structure.
The set $\{1,\ldots, N\}$ is denoted by $[N]$ and the vector $[1,\ldots 1] \in \reals^{N}$ by $\ones_N$.

\subsection{MDP Congestion game}
In an archetypal finite MDP problem, each decision maker solves a finite-horizon MDP~\cite{altman1999constrained} with horizon length $T$, state space $[S]$, and action space $[A]$ given by
\begin{equation}\label{eqn:individualCG}
\begin{aligned}
\underset{x_{sa}}{\min} & \ \underset{s \in [S]}{\sum}\underset{a \in [A]}{\sum} x_{sa}c_{sa} \\
\text{s.t.} &\  \sum_s \sum_a x_{sa} = 1, \\
&\ \sum_{a} x_{sa} = \underset{s' \in [S]}{\sum}\underset{a\in[A]}{\sum}P_{ss'a}, \ \forall \ s \in [S], \\ 
& \ x_{sa} \geq 0,  \quad \forall \ s \in [S], a \in [A],
\end{aligned}
\end{equation}
where the objective is to minimize the expected average cost over an infinite time horizon with a finite set of actions $[A]$ and a finite set of states $[S]$.
The optimization variable $x\in \reals_+^{SA}$ defines a state-action distribution of an individual decision maker, such that $x_{sa}/\sum_{a'\in[A]}x_{sa'}$ denotes a decision maker's probability of taking action $a$ at state $s$.

The probability kernel $P \in \reals_+^{S \times SA }$ has form 
\[P = \begin{pmatrix}P_{s_1s_1a_1} &  P_{s_1s_1a_2} & \ldots &P_{s_1s_2a_1} & \ldots & P_{s_1s_Sa_A} \\
P_{s_2s_1a_1} &  P_{s_2s_1a_2} & \ldots&P_{s_2s_2a_1} & \ldots & P_{s_2s_Sa_A} \\
\vdots \\
P_{s_ns_1a_1} &  P_{s_ns_1a_2} & \ldots&P_{s_ns_2a_1} & \ldots & P_{s_ns_Sa_A}
\end{pmatrix}, \]
where $P_{ss'a}$ denotes the transition probability from state $s'$ to $s$ when taking action $a$. $P$ is column stochastic and  defines the transition dynamics.

In a non-atomic MDP congestion game, an infinite number of decision makers each solves an MDP on the same state-action space. The total population distribution is described by $y \in \reals^{SA}_+$. 

\begin{assumption}[Mean Field Assumption]\label{ass:meanField}
In the limit where the number of decision makers approaches to infinity, the total population becomes a continuous distribution $y \in \reals_+^{SA}$ with total mass $M > 0$, where $y_{sa}$ denotes the portion of population who chooses action $a$ at state $s$.
\end{assumption}
The \emph{population distribution} $y$ relates to  individual state-action distribution by $y = \sum_k \alpha_k x(k),\,\,
    \sum_{k\in\mc{K}}\alpha_k = M, \,\, \alpha_k > 0, \,\, \forall \ k \in \mc{K}, $
where $\mc{K}$ is the index set of feasible distributions with respect to MDP~\eqref{eqn:individualCG}, and $\alpha_k$ corresponds to the portion of population that chooses distribution $x(k)$.

Assumption~\ref{ass:meanField} results in a \emph{non-atomic} nature of MDP congestion games: each decision maker's state-action distribution is infinitesimal with respect to the population distribution, and changes in an individual $x$ does not affect $y$. 

In an MDP congestion game, the state-action costs $c_{sa}$ are population dependent functions, i.e.,
$c_{sa} = \ell_{sa}(y_{sa})$, where $\ell_{sa}: \reals_+ \rightarrow \reals$.
We denote $\ell: \reals^{SA}_+ \rightarrow \reals^{SA}$ as the vector of state-action costs. The population dependency of $\ell$ reflects \emph{congestion effects}: the greater the population in a given state-action pair, the greater the cost of taking that state-action for all decision makers. This assumption is consistent with practical networked interactions in traffic and telecommunications~\cite{bureau1964traffic} where, e.g., the cost of traversing a road increases for each driver when the number of cars on the road increases. 
\begin{assumption}
\label{ass:increasing}
The state-action costs $\ell: \reals^{SA}_+ \rightarrow \reals^{SA}$ are continuously differentiable and $\nabla_y\ell$ is positive definite.
\end{assumption}
In an MDP congestion game, all decision makers achieve their optimal expected cost when the population distribution is at \emph{MDP Wardrop equilibrium}.
\begin{definition}[MDP Wardrop Equilibrium~\cite{calderone2017infinite}] \label{def:VI_mdpWE}
A population distribution $y^\star$ which satisfies Assumption~\ref{ass:meanField} is a Wardrop equilibrium when each decision maker's probability $x^\star(k)$ satisfies
\[\sum_{s\in[S]}\sum_{a\in[A]}\ell_{sa}(y^\star_{sa})(x^\star(k)_{sa} - x_{sa}) \leq 0, \  \forall k \in \mathcal{K}.\]

\end{definition}
Definition~\ref{def:VI_mdpWE} defines optimality in terms of expected cost: an individual decision maker deviating from its current strategy will not achieve a more optimal expected cost. 
\subsection{Directed Hypergraphs}
Similar to stochastic shortest path problems~\cite{epstein2009efficient}, MDP congestion game is inherently related to hypergraphs~\cite{gallo1993directed}. 
We consider a weighted directed hypergraph $\G = ([S], \E)$, where $[S]$ is the set of states considered in MDP congestion game and $\E$ is the set of \emph{hyperarcs}. A \emph{hyperarc} $(s,a)$ is defined for each state-action pair, such that the tail is always at $s$, and the head, $\mc{H}(s,a)$, is the set of states that can be reached from state $s$ taking action $a$---i.e., 
$\mc{H}(s,a) = \{s' \in [S]\ |\ P_{s'sa} > 0\}$. 
\begin{figure}[h]
    \centering
    \includegraphics[width=0.45\columnwidth]{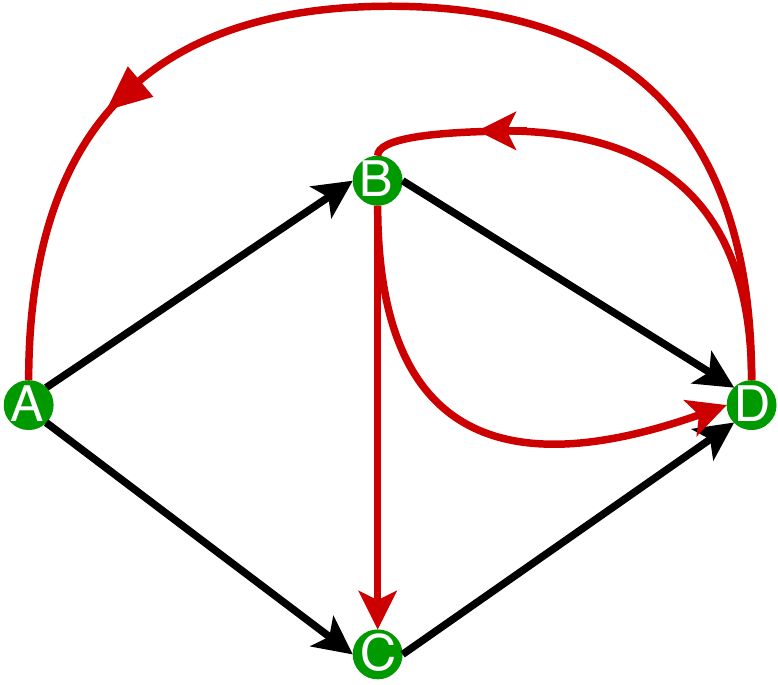}
    \caption{A directed hypergraph with 4 states. The hyperarcs in red have one tail but multiple heads, denoting possible states $s'$ that taking state-action $(s,a)$ may result in next.}
    \label{fig:hypergraph}
\end{figure}

A \emph{hypergraph incidence} matrix $E \in \reals^{S\times |\E|}$ has elements defined as 
\begin{equation}\label{eqn:weightedIncidence}
    (E)_{s', (s,a) } = \begin{cases}
        1 &   s' = s,\\
        -P_{s'sa} &  s' \neq s.
    \end{cases}
\end{equation}
Alternatively the incidence matrix can be written as $E = (I_S \otimes \ones_A^T - P)$.
In this form, we can see that the difference in probability density per state (i.e., $(I_S \otimes \ones_A^T) x$) before and after a stochastic transition (i.e., $Px$) can be written as $Ex$. Therefore a stationary distribution $\hat{x}$ always satisfies $E\hat{x} = 0$. 


A directed hypergraph is \emph{strongly connected} if every non-empty subset $\mc{R}\subset [S]$ has at least one incoming hyperarc from the set $[S]/\mc{R}$. 
In the following  consider hypergraphs whose incidence matrix has rank $S-1$. 
\begin{assumption}[Incidence Rank]
\label{ass:incidenceRank}
The hypergraph that corresponds to probability transition kernel $P$ is strongly connected, and its incidence matrix $E$ has row rank $S - 1$. 
\end{assumption}

An MDP congestion game can be stated as an optimization problem over population distribution $y$~\cite{calderone2017infinite}, formulated as
\begin{subequations}
\begin{align}
\underset{y}{\min} & \quad \sum\limits_{s\in[S]} \sum\limits_{a\in[A]} \int_0^{y_{sa}}\ell_{sa}(u)du\\
\mbox{s.t.} & \quad Ey = 0,\label{eqn:mdpcg_conservation}\\ 
& \quad \ones^Ty = M,\label{eqn:mdpcg_totalMass}\\
& \quad y \geq 0,  \label{eqn:mdpcg_positivity}
\end{align}
\label{eqn:mdpgame}
\end{subequations}
where constraints on $y$ is derived from feasibility conditions of individual decision makers.

Let $\nu$, $\lambda$, $\mu$  be Lagrange multipliers corresponding to~\eqref{eqn:mdpcg_conservation}, \eqref{eqn:mdpcg_totalMass}, \eqref{eqn:mdpcg_positivity}, respectively. 
When $\ell$ satisfies Assumption~\ref{ass:increasing}, uniqueness of the tuple $(y^\star,\lambda^\star, \mu^\star)$ is guaranteed~\cite{calderone2017infinite}. However due to the rank deficiency of $E^T$, $\nu^\star$ must be non-unique. We show next that the constraint $Ey = 0$ can be reduced to a full ranked condition, such that the corresponding optimal Lagrange multiplier $\nu^\star$ is unique.  
\begin{lemma}[Full Row Rank Incidence Matrix]
\label{lem:fullRankIncidence}
An MDP congestion game~\eqref{eqn:mdpgame} that satisfies Assumption~\ref{ass:incidenceRank} is equivalent to 
\begin{equation}
\begin{aligned}
\underset{y}{\min} & \quad \sum\limits_{s\in[S]} \sum\limits_{a\in[A]} \int_0^{y_{sa}}\ell_{sa}(x)dx \\
\mbox{s.t.} & \quad \Tilde{E}y = 0,\\ 
& \quad \ones^Ty = M,\\
& \quad y \geq 0, \\
\end{aligned}
\label{eqn:fullRankMdpGame}
\end{equation}
where $E = \begin{bmatrix}\Tilde{E}\\ e^T\end{bmatrix}$ and $\Tilde{E}$ has full row rank.
\end{lemma}
\begin{proof}
Consider removing arbitrary row vector $e^T$ from the incidence matrix $E$. By Assumption~\ref{ass:incidenceRank}, $e^T$ is not identically $0$. 
Clearly, $Ey = 0$ implies $\Tilde{E}y = 0$. To see that the opposite implication, note that $E^T\ones = 0$ from definition leads to $\ones^T\Tilde{E} = -e^T$. Therefore $\Tilde{E}y = 0$ implies $-e^Ty = 0$.\end{proof}
The Karush-Kuhn-Tucker (KKT) conditions of~\eqref{eqn:fullRankMdpGame} are
\begin{equation}\label{eqn:kkt_mdpcg}
\begin{aligned}
H(y^\star, \nu^\star, \lambda^\star, \mu^\star) = 
\begin{bmatrix}
\ell(y^\star) - \Tilde{E}^T\nu^\star -\lambda^\star\ones- \mu^\star &\\
\Tilde{E}y^\star&\\
\ones^Ty^\star - M&\\
(\mu^\star)^T y^\star &\\
\end{bmatrix}& = \begin{bmatrix}0\\0\\0\\ 0\end{bmatrix}, \\
\mu^\star \geq 0, y^\star & \geq 0. \\
\end{aligned}
\end{equation}
where $\nu \in \reals^{S-1}$, $\lambda \in \reals$, $\mu \in \reals_+^{SA}$ are uniquely determined for a given population distribution $y$.
\section{Sensitivity Analysis}
\label{sec:sensitivityAnalysis}
In this section, we derive a sensitivity characterization of stochastic Braess paradox. 
To facilitate the analysis, we introduce \emph{perturbation dependent} cost functions $\ell: \reals^{SA} \times\reals^{SA} \rightarrow \reals^{SA}$ that is continuously differentiable in both inputs, where the additional input represents perturbation to the cost function. The game itself is played with respect to a given perturbation $\epsilon$ and a corresponding cost $\ell(\cdot,\epsilon)$. 

The KKT conditions~\eqref{eqn:kkt_mdpcg} can also be viewed as an implicit characterization of optimal population $y^\star$ as parameterized by $\epsilon$. We define a \emph{point-to-set} mapping given by
\begin{equation}\label{eqn:point2SetMapping}
    Q : \epsilon \mapsto \left\{ (y,\nu,\lambda,\mu) | \ H(\epsilon, \lambda, \nu, y, \mu) = 0 ,\ \mu \geq 0,\ y\geq 0\right\} .
\end{equation}
The point-to-set mapping, $Q(\epsilon)$, generalizes local differentiability of $y^\star$ as a function of $\epsilon$~\cite{dontchev2009implicit}. For an $\epsilon$, if the optimal distribution $y^\star$ and corresponding optimal Lagrange multipliers are unique,  $Q(\epsilon)$ is a \emph{single valued set mapping}; in this  case we denote the optimal population distribution by $y^\star(\epsilon)$. Unless otherwise stated, Assumption~\ref{ass:increasing} holds from now on. 

Consider an MDP congestion game played with costs $\ell(y,0)$ and its optimal solution $y^\star(0)$. When $Q(\epsilon)$ is a single valued set mapping for $\epsilon$ in an open set containing zero, the Jacobian $\nabla_\epsilon y^\star(0)$ exists. We call  $\nabla_\epsilon y^\star(0)$ the \emph{sensitivity} of MDP Wardrop equilibria---i.e.,how $y^\star(0)$ changes when cost $\ell$ is perturbed by $\epsilon$.  

We restrict our attention to MDP congestion games whose unique equilibrium satisfies $y^\star(0) >0 $. 
\begin{assumption}[Positivity Condition]\label{ass:positivity}
The optimal population distribution of the unperturbed MDP congestion game satisfies $y^\star > 0$. 
\end{assumption}
Assumption~\ref{ass:positivity} is not restrictive in the following sense: when state-action costs satisfy Assumption~\ref{ass:increasing}, Assumption~\ref{ass:positivity} will always be satisfied for some total mass $M> 0$. Consider cost functions that satisfy $\ell_{sa}(0) = b_{sa} \in \reals$. If a hyperarc is not optimal, i.e. has no mass, then $b_{sa}$ must be at least $\max_{a'\in [A]} \ell_{sa'}(y^\star_{sa'},0)$. However, all other state action costs must increase as total mass $M$ increases, therefore a total mass threshold exists for which $\max_{a'\in [A]} \ell_{sa'}(y^\star_{sa'},0) \geq b_{sa}$, past which $(s,a)$ will become optimal.
\begin{proposition}[Perturbation Map]
\label{pro:P2Smapping}
If an MDP congestion game~\eqref{eqn:mdpgame} satisfies Assumptions~\ref{ass:increasing} and~\ref{ass:incidenceRank} with costs $\ell(\cdot,\epsilon)$, and $y^\star(\epsilon)$ satisfies Assumption~\ref{ass:positivity}, then the mapping $Q(\epsilon)$~\eqref{eqn:point2SetMapping} is a single valued mapping at $\epsilon$.
\end{proposition}
\begin{proof}
From Assumptions~\ref{ass:increasing} and~\ref{ass:positivity}, there exists a unique $y^\star(\epsilon) > 0$ solving the KKT conditions~\eqref{eqn:kkt_mdpcg} for costs $\ell(\cdot, \epsilon)$. Lagrange multiplier $\mu^\star = 0$ from complementary slackness. The other optimal solutions can be determined by solving $(y^\star)^T(\ell(y^\star) - \Tilde{E}^T\nu^\star - \lambda^\star\ones) = 0$, which implies $\lambda^\star = (y^\star)^T\ell(y^\star)/M $. Furthermore, unique $y^\star$ and $\lambda^\star$ implies $\Tilde{E}\nu^\star$ is unique. Since $\Tilde{E}^T$ has full rank, $\nu^\star$ is unique.
\end{proof}
Proposition~\ref{pro:P2Smapping} implies that when $\ell$ is continuously differentiable at $y^\star$ and $\epsilon = 0$, there exists a continuously differentiable and invertible function of the optimal distribution $y$ in terms of $\epsilon$. We note that similar sensitivity results which do not consider stochastic congestion effects exist for routing games~\cite{Tobin1988}. However, our results for MDP congestion games are less restrictive due to the lack of the dual route/link space. 


\begin{theorem}
[MDP Congestion Game Flow Sensitivity]\label{thm:sensitivity}
Consider an MDP congestion game with costs $\ell(y, \epsilon)$, such that $\ell$ is a continuously differentiable function of $(y,\epsilon)$ and satisfies Assumption~\ref{ass:increasing}, and the associated hypergraph satisfies Assumption~\ref{ass:incidenceRank}.  
If the optimal population distribution $y^\star(\epsilon^\star) > 0$, the sensitivity of the MDP Wardrop equilibrium is given by
\[ \nabla_{\epsilon} y^\star = G^{-1}N(N^TG^{-1}N)^{-1}N^TG^{-1}J - G^{-1}J. \]
Moreover, the sensitivity of optimal state-action costs is
 \[\nabla_{\epsilon}\ell(y^\star,\epsilon^\star) = N(N^TG^{-1}N)^{-1}N^TG^{-1}J,\]
where $N = \begin{bmatrix}\Tilde{E}^T & \ones \end{bmatrix}$, $\Tilde{E}$ as given by Lemma~\ref{lem:fullRankIncidence}, $G = \nabla_y \ell(y^\star(\epsilon^\star), \epsilon^\star)$, and $J = \nabla_\epsilon \ell(y^\star(\epsilon^\star), \epsilon^\star)$.  
\end{theorem}

\begin{proof}
From Proposition~\ref{pro:P2Smapping}, the game with costs $\ell(\cdot, \epsilon)$ has associated single valued mapping $Q(\epsilon)$ in a neighborhood of $\epsilon^\star$, then $H(Q(\epsilon), \epsilon) = 0$ implies the total derivative $d H(Q(\epsilon), \epsilon)/d \epsilon = 0 $ for $\norm{\epsilon - \epsilon^\star} \leq \delta$. Let $w = \begin{pmatrix}y & \nu & \lambda\end{pmatrix}$ and $f(y, \nu, \lambda, \epsilon) = 
\begin{bmatrix}
\ell(y, \epsilon) - \Tilde{E}^T\nu -\lambda\ones &\\
\Tilde{E}y&\\
\ones^Ty - M&
\end{bmatrix}$.
Like $H$, $f$ is continuously differentiable in $w$, and is equal to $0$ at $y^\star(\epsilon)$ and corresponding optimal Lagrange multipliers. From the implicit  function theorem~\cite[Sec.1B]{dontchev2009implicit}, when $\nabla_{w}f(w, \epsilon^\star)$ is invertible, 
$\nabla_\epsilon w^\star =  \big(\nabla_{w}f(w^\star, \epsilon^\star)\big)^{-1} \nabla_\epsilon f(w^\star, \epsilon^\star)$.
We wish to show that
$ \nabla_{w} f(w^\star, \epsilon^\star) = \begin{pmatrix} G & -N \\ N^T & 0\end{pmatrix}$ is non-singular.
The Schur complement of $\nabla_v f(Q^\star(\epsilon), \epsilon^\star)$ with respect to the lower block diagonal component $0$ is $ N^TG^{-1}N$. From Assumptions~\ref{ass:incidenceRank} and~\ref{ass:increasing}, $N^T$ has full row rank and $G \succ 0$. Therefore $N^TG^{-1}N$ is positive definite and non-singular and equivalently, $\nabla_{w} f(w^\star, \epsilon^\star) \succ 0$ and non-singular. 

The partial gradient of $f(w^\star, \epsilon^\star)$ with respect to $\epsilon$ is
\[\nabla_\epsilon f(w^\star, \epsilon^\star)= \begin{pmatrix} J  & 0 \\ 0 & 0\end{pmatrix}.\]
We use Gaussian elimination to invert $ \nabla_{w} f(w^\star, \epsilon^\star)$ and get 
\[( \nabla_{Q(\epsilon)} f(Q^\star(\epsilon), \epsilon^\star) )^{-1} = \begin{pmatrix}
A & B \\ 
C & D \\
\end{pmatrix}.\]
where $A= G^{-1} - G^{-1}N(N^TG^{-1}N)^{-1}N^TG^{-1}$, $B=  - G^{-1} N(N^TG^{-1}N)^{-1}$, $C = B^T$, $D =  (N^TG^{-1}N)^{-1}$.

We decompose $w$ to its components and solve for $\nabla_\epsilon y^\star$,
\begin{align*}
   & \nabla_{\epsilon}\begin{bmatrix}y^\star\\\nu^\star\\\lambda^\star\end{bmatrix}= -
    \begin{pmatrix}
    G^{-1}(J - N(N^TG^{-1}N)^{-1}N^TG^{-1}J) & 0 \\
    -(N^TG^{-1}N)^{-1}N^TG^{-1}J & 0 \\
    \end{pmatrix},
\end{align*}
where the first row corresponds to $\nabla_{\epsilon}y^\star(\epsilon^\star)$ and the second row corresponds to $ \nabla_{\epsilon} \begin{bmatrix} \nu^\star & \lambda^\star\end{bmatrix}^T $. The first block corresponds to $\nabla_{\epsilon} y^\star(\epsilon^\star)$. Note that because $y^\star(\epsilon^\star) > 0$, we can express the optimal cost as 
\[\ell^\star = \begin{bmatrix}\Tilde{E}^T & \ones \end{bmatrix}\begin{bmatrix}\nu^\star\\\lambda^\star\end{bmatrix} = N\begin{bmatrix}\nu^\star\\\lambda^\star\end{bmatrix}. \]
The sensitivity of the costs $\ell^\star$ with respect to perturbation is 
\[\nabla_{\epsilon}\ell^\star = N \nabla_{\epsilon}\begin{bmatrix}\nu^\star\\\lambda^\star\end{bmatrix} = N(N^TG^{-1}N)^{-1}N^TG^{-1}J.\]
\end{proof}

\subsection{Stochastic Braess Paradox}
In the routing game literature, a well-known phenomenon that is related to the sensitivity of optimal routes is \emph{Braess paradox}~\cite{braess1968paradoxon}. The phenomenon refers to the paradoxical effect that occurs when costs of traversing edges are \emph{decreased}, resulting in an increase in player's average cost. We show that a similar behaviour exists in MDP congestion games, and its occurrence can be linked to the underlying hypergraph structure through sensitivity analysis. Consider the \emph{social cost} of an MDP congestion game,
$J(y, \ell) = y^T\ell(y)$. 

Stochastic Braess paradox can be defined by the sensitivity of the social cost of MDP congestion games.

\begin{definition}[Stochastic Braess Paradox]
For two MDP congestion games~\eqref{eqn:mdpgame} satisfying Assumption~\ref{ass:increasing} defined on the same hypergraph, their respective costs $\ell$ and $\bar{\ell}$ satisfies
\[\ell(y) - \bar{\ell}(y) \geq 0,\quad \forall \{y \ |\ Ey = 0,\ \ones^Ty = M,\ y \geq 0\}.\] 
Let the optimal population distribution be $y^\star$ and $\bar{y}^\star$, respectively. A \emph{stochastic Braess paradox} occurs when the social cost
satisfies $J(y^\star, \ell) < J(\bar{y}^\star, \bar{\ell})$.
\end{definition}

When $\ell$ and $\bar{\ell}$ are instantiated by different $\epsilon$ values of the same continuously differentiable function $\ell(\cdot, \epsilon)$, the existence of Braess paradox suggests that there is a perturbation which increases the state-action costs from $\bar{\ell}$ to $\ell$ such that $ J(y^\star, \ell)< J(\bar{y}^\star, \bar{\ell}) $.   
\begin{corollary}[Sufficient Conditions for stochastic BP]
Consider a feasible MDP congestion game~\eqref{eqn:mdpgame} which satisfies Assumptions~\ref{ass:increasing} and  \ref{ass:incidenceRank} with an optimal population distribution $y^\star > 0$. Its \emph{social cost sensitivity} can be defined as
\begin{equation}
\begin{aligned}
   \nabla_{\epsilon} J = &( G^{-1}N(N^TG^{-1}N)^{-1})^{-1}N^TG^{-1} - G^{-1})\ell(y^\star)\\
    &  +  N(N^TG^{-1}N)^{-1}N^TG^{-1}y^\star.
    \label{eqn:localSensitivity}
\end{aligned}    
\end{equation}
Then, $\nabla_{\epsilon} J  \notin \reals_+^{|\mc{S}| |\mc{A}|}$ is a sufficient condition for the occurrence of \emph{stochastic Braess paradox}.
\end{corollary}
\begin{proof}
$J$ is bilinear and therefore continuously differentiable in $\ell^\star$ and $y^\star$. From Theorem~\ref{thm:sensitivity}, there exists a neighbourhood $\norm{\epsilon} \leq \delta$ within which $J$ is continuously differentiable in $\epsilon$, and the Jacobian is given as
\[\nabla_\epsilon J(\ell^\star, y^\star) = \nabla_{y^\star}J\nabla_{\epsilon}y^\star + \nabla_{\ell^\star}J\nabla_{\epsilon}\ell^\star.\]
For any $\nabla_\epsilon J \notin \reals^{|\mc{S}| |\mc{A}|}_+$, there exists $\epsilon \in \reals^{|\mc{S}| |\mc{A}|}_+$ such that $\norm{\epsilon} \leq \delta$ and $\epsilon^T \nabla_\epsilon J < 0$.
We then consider the MDP congestion game with costs $\bar{\ell}$ and equilibrium $\bar{y}^\star$, where $\bar{\ell}$ is defined by 
\[\bar{\ell} = \ell + \epsilon.\]
By the mean value theorem, there exists $k \in (0,1]$ where 
\[J(\bar{y}^\star, \bar{\ell}^\star) = J(y^\star, \ell^\star) +(k\epsilon)^T\nabla_{\epsilon}J. \]
Since $k\epsilon^T\nabla_{\epsilon}J(\delta) < 0$, $J(\bar{y}^\star, \bar{\ell}) < J(y^\star, \ell^\star)$ holds.
\end{proof} 

\section{Role of Stochasticity}
\label{sec:stochasticity}
In this section, we consider the deterministic counterpart of MDP congestion games to evaluate how the introduction of \emph{stochasticity} influences social cost sensitivity.
\subsection{Cycle Game}\label{sec:cycleGame}
A directed \emph{primal graph}~\cite{adler2007hypertree} $\mc{G}_d = ([S], \E_d)$ can be derived from a hypergraph $\mc{G}= ([S], \E)$, by considering the same set of states and define \emph{edge set} $\E_d$ defined by 
\[e = (s_1, s_2) \in \E_d \text{ if } \exists \ (s_1, a) \text{ s.t. } P_{s_2s_1a} > 0. \]
Its incidence matrix $D \in \reals^{S\times\mc{E}_d}$ is given by
\[
[D]_{ie} = \begin{cases}
\ \ 1, & \text{ if edge } $e$ \text{ starts at state }$i$, \\
-1, & \text{ if edge } $e$ \text{ ends at state }$i$, \\
\ \ 0, & \text{ otherwise. } \\
\end{cases}
\]
An MDP congestion game~\eqref{eqn:mdpgame} can be played on $\mc{G}_d$ for a given cost $\ell$. The constraint $Dy = 0$ implies that any feasible population distribution must be a combination of cycles of $\mc{G}_d$~\cite{godsil2001cuts}. Therefore, we call a deterministic MDP congestion game where all state-action pairs lead to deterministic outcomes, a \emph{cycle game}~\cite{calderone2017infinite}. 

The edge set of a primal graph dictates \emph{allowable} transitions over state space $[S]$, where as a hyperarc corresponds to a discrete set of particular probability distributions assignments to these allowable transitions as given by $\mc{E}_d$.
We consider a transformation $T \in \reals_+^{|\mc{E}_d|\times |\E|}$ between the incidence matrix of a hypergraph $E$ and its host graph $D$, such that $E = DT$. Columns of $T$ denote how an action $a$ distributes mass over edges adjacent to $s$ of the primal graph, 
\begin{equation}
   T_{(s_1,s_2),(s,a)} = \begin{cases}
   P_{s_2as}, & s_1 = s, \\
   0, & \text{otherwise.}
   \end{cases}
\end{equation}
In addition to being element-wise non-negative, $T$ is also column stochastic---i.e.,
\[\underset{e\in\mc{E}_d}{\sum}T_{e,(s,a)} = \underset{s' \in \mc{S}}{\sum}P_{s'as} = 1.\]
An example is given in Fig.~\ref{fig:cycleGame} in which labeled edges are defined between states $\{A, B, C\}$. The incidence and transformation matrices corresponding to Fig.~\ref{fig:cycleGame} is given by
\[D = \begin{bmatrix}
0 & -1 & 0& 1 \\
1 & 1 & -1 & 0 \\
-1 & 0 & 1 & -1
\end{bmatrix}, T = \begin{bmatrix}
0.4 & 0 & 0 & 0\\
0.6 & 1 & 0 & 0\\
  0 & 0 & 1 & 0\\
  0 & 0 & 0 & 1
\end{bmatrix}.\]
\begin{figure}[h]
\centering
\includegraphics[width=0.3\columnwidth]{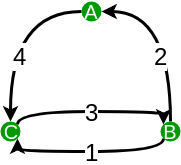}
\caption{Example graph structure of a cycle game.}
\label{fig:cycleGame}
\end{figure}
The eigenvalues of $T$ characterize the {amount} of stochasticity introduced by the MDP dynamics. When $T = I$, the MDP congestion game is itself a cycle game with no stochasticity. When each state-action pair uniformly distributes the probability over available edges, $T$ has a block diagonal structure with eigenvalues less than 1 if a state has two or more actions available. 
Fig.~\ref{fig:cycleGame} also provides an example of a feasible transformation $T$ that is invertible.
\subsection{Effects of Stochasticity}
When the incidence matrix of a hypergraph is related to the incidence matrix of the corresponding primal graph by an invertible transformation $T$, there is a direct relationship between the equilibria of the MDP congestion game and cycle game played on these graphs. 
\begin{assumption}[Invertible Transformation $T$]
\label{ass:invertibleCycleGame}
A directed hypergraph $\mc{G} = ([S], \E)$ can be induced from its directed primal graph $\mc{G}_d = ([S], \mc{E}_d)$, such that $|\mc{E}| = |\mc{E}_d| $, and the incidence matrices, $E$ and $D$, of the two graphs, respectively, are related by an invertible transformation $T$.
\[E = DT, \qquad T \in \reals_+^{|\mc{E}_d| \times |\mc{E}| },\  \ones^TT = \ones^T.\]
\end{assumption}

\begin{proposition}[Equilibria Relationship]
If the graph $\mc{G}$ of an MDP congestion game satisfies Assumption~\ref{ass:invertibleCycleGame}, $y^\star > 0 $ is an MDP Wardrop equilibrium if and only if $Ty^\star$ is an equilibrium of the cycle game defined on $\mc{G}_d$ with costs $\ell_e$ on its edges where
\[\ell_e(\cdot) = T^{-T} \ell_{sa}\circ T^{-1}(\cdot).\]
\end{proposition}
\begin{proof}
Consider an MDP Wardrop equilibrium $y^\star$ that satisfies Assumption~\ref{ass:positivity}, then there exists primal variable solution $y^\star$ and dual variables $\nu^\star$, $\lambda^\star$ that satisfy the KKT conditions~\eqref{eqn:kkt_mdpcg} with $\mu^\star = 0$. We can re-write $H(y,\nu,\lambda,\mu) = 0$ from~\eqref{eqn:kkt_mdpcg} with transformations $DT = E$ and $z^\star = Ty^\star$, and $\mu^\star = 0$,
\begin{equation}\label{eqn:cycle_kkt}
\begin{aligned}
T^{-T}\ell(T^{-1} z^\star) - D^T\nu^\star - \lambda^\star T^{-T}\ones & = 0,\\
Dz^\star & = 0, \\
\ones^TT^{-1}z^\star - M & = 0.
\end{aligned}
\end{equation}
Since $T$ is element-wise non-negative, and $y^\star > 0$, $Ty^\star = z^\star > 0$. By construction, $T^{-1}$ is column stochastic, therefore $T^{-T}\ones = \ones$. Therefore~\eqref{eqn:cycle_kkt} is equivalent to the KKT conditions of a game with cost $T^{-T}\circ \ell \circ T^{-1}$, deterministic incidence matrix $D$, and optimal population distribution $z^\star$. 

We note that $T^{-T}(\nabla \ell) T^{-1}$ is positive definite, and while an individual state-action cost $(T^{-T}\circ \ell \circ T^{-1})_{sa}$ requires \emph{multiple} hyperarcs' population distribution to define the congestion cost at $(s,a)$, it defines a potential game~\cite{calderone2017markov} consistent with Assumption~\ref{ass:increasing}. This implies that~\eqref{eqn:cycle_kkt} coincides with the KKT conditions of a cycle game formulation with costs $T^{-T}\circ \ell \circ T^{-1}$, incidence matrix $D$, and mass $M$. Since $z^\star > 0$ satisfies the KKT conditions of this cycle game, $z^\star$ is the cycle game's unique optimal population distribution.
\end{proof}
The relationship between the equilibria of the deterministic game and the equilibria of the game allows for a direct comparison between the sensitivity of the social cost in the two games. We show next that the social cost sensitivity of a MDP congestion game can be directly bounded by the eigenvalues of $T$, ie the amount of stochasticity introduced.  
\begin{theorem}[Effects of Stochasticity]
\label{thm:effect_stochasticity}
We consider an MDP congestion game~\eqref{eqn:mdpgame} and a cycle game (Section ~\ref{sec:cycleGame}) whose graphs satisfy Assumption~\ref{ass:incidenceRank}. Let the social cost of the cycle game be $J_c$, and the social cost of the MDP congestion game be $J$, the sensitivity of the cycle game can be bounded by
\[\norm{\nabla_\epsilon J_c}_2 \leq \norm{T}_2 \norm{\nabla_\epsilon J}_2. \]
\end{theorem}
\begin{proof}
Let $N_c = \begin{bmatrix}
\bar{D}^T & \ones
\end{bmatrix}$, where $\bar{D}$ is $D$ with any one row removed. From Assumption~\ref{ass:incidenceRank}, the removed row cannot be identically zero as that would ensure $\rank(D) \leq S - 2$, then $N_c$ is related to $N = \begin{bmatrix}\bar{E}^T & \ones\end{bmatrix}$ by $T^TN_c = N$ where $\bar{E}$ has the same row removed.

Since $z^\star = Ty^\star$, the sensitivity of the cycle game social cost $J_c =  (z^\star)^T T^{-T}\ell(T^{-1}z^\star)$ can be evaluated at $(y^\star, \ell^\star)$,
\[\nabla_\epsilon J_c\begin{pmatrix}y^\star \\ \ell(y^\star)\end{pmatrix} = \begin{pmatrix}
 T^{-T}A T^TT & 0 \\
 0 & TB
\end{pmatrix}\begin{pmatrix}y^\star \\ \ell(y^\star)\end{pmatrix}.\]
where $A = N(N^TG^{-1}N)^{-1}N^TG^{-1}$ and $B = G^{-1} - G^{-1}N(N^TG^{-1}N)^{-1})^{-1}N^TG^{-1}$. 
In comparison, the sensitivity of the MDP congestion game's social cost is
\[\nabla_\epsilon J\begin{pmatrix}y^\star \\ \ell(y^\star)\end{pmatrix} = \begin{pmatrix}
 A& 0 \\
 0 & B
\end{pmatrix}\begin{pmatrix}y^\star \\ \ell(y^\star)\end{pmatrix}.\]
We can compare the social cost sensitivity Jacobian for the cycle game and the MDP congestion game, denoted by $M_c$ and $M$ respectively. 
\begin{equation}
\begin{aligned}
  \norm{M_c}_2 & = \sigma_{max}\{T^{-T}AT^TT, TB\} \\
  & \leq \norm{T}_2 \norm{M}_2.
\end{aligned}
\end{equation}
\end{proof}
Theorem~\ref{thm:effect_stochasticity} states that given equivalent Wardrop equilibria, the sensitivity of the social cost in the deterministic cycle game is always bounded by the sensitivity of the MDP congestion game and the amount of stochasticity introduced. Since $\norm{T}_2 \leq 1$, Theorem~\ref{thm:effect_stochasticity} states that introducing stochasticity \emph{increases} effects of Braess paradox.
\section{Simulations}\label{sec:wheatstoneExamples}
In this section, we use the results of sensitivity analysis on a hypergraph derived from a directed Wheatstone graph. Wheatstone structure is known to induce Braess paradox for non-atomic routing games~\cite{milchtaich2006network}, we analyze its behaviour under stochastic transitions and show that not only does stochastic Braess paradox also occur, but we can avoid the paradox by our sensitivity analysis. We demonstrate Theorem~\ref{thm:sensitivity} by cost perturbations in both the negative and positive directions of the social cost sensitivity, and validating the predictions with simulated results. 
\begin{figure}[h]
    \centering
    \includegraphics[width=0.45\columnwidth]{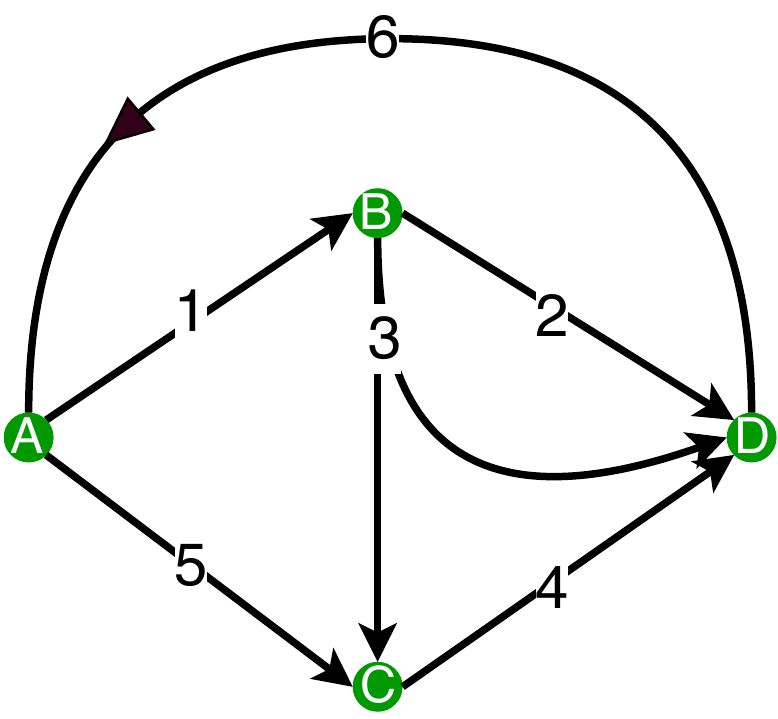}
    \caption{Hypergraph structure of MDP congestion game}
    \label{fig:MDPNetwork}
\end{figure}

Consider an MDP congestion game defined on hypergraph shown in Figure~\ref{fig:MDPNetwork}. We play the MDP congestion game defined by \eqref{eqn:mdpgame}, with a scaled mass $M = 1$. The cost functions are defined as $\ell_{sa}(y_{sa}) =  A_{sa} y_{sa} + b_{sa}$.
\begin{table}[h]
\centering
\begin{tabular}{@{}lll}
\toprule
         & $A_{sa}$       & $b_{sa}$                          \\
$\ell_1$ & 9              & 1           \\
$\ell_2$ & 0.1            & 1               \\
$\ell_3$ & 0.1            & 0               \\
$\ell_4$ & 9              & 1               \\
$\ell_5$ & 0.1            & 0.1           \\
$\ell_6$ & 0.1            & 0               \\ \bottomrule
\end{tabular}
\caption{Distribution dependent hyperarc costs}
\label{tab:linkCosts}
\end{table}

All state-action pairs correspond to hyperarcs, but all state-action pairs except for hyperarc $3$ define deterministic actions. The stochastic incidence matrix is defined by
\[E = \begin{pmatrix}
1& 0 & 0 & 0 &1 & -1 \\
-1 & 1 & 1 & 0 & 0 & 0 \\
0 & 0 & -0.9 & 1 & -1 & 0 \\
0 & -1 & -0.1& -1& 0 & 1 \\
\end{pmatrix}.\]
Note that when a hyperarc has one head state, its corresponding column of incidence matrix $E$ is identical to that of the cycle game incidence matrix $D$ (Section~\ref{sec:cycleGame}). Stochastic hyperarcs are convex combinations of the deterministic edges that correspond to allowable state transitions originating from the same tail state. 

We simulate each MDP congestion game by solving the convex optimization formulation given by~\eqref{eqn:mdpgame} with cvxpy. First, we verify in Figure~\ref{fig:mdpPositiveOptimalDistribution} that at given costs $\ell$, the optimal population distribution $y^\star$ is strictly positive. 
\begin{figure}[ht]
    \centering
    \includegraphics[trim=12 40 15 30, width=0.5\columnwidth]{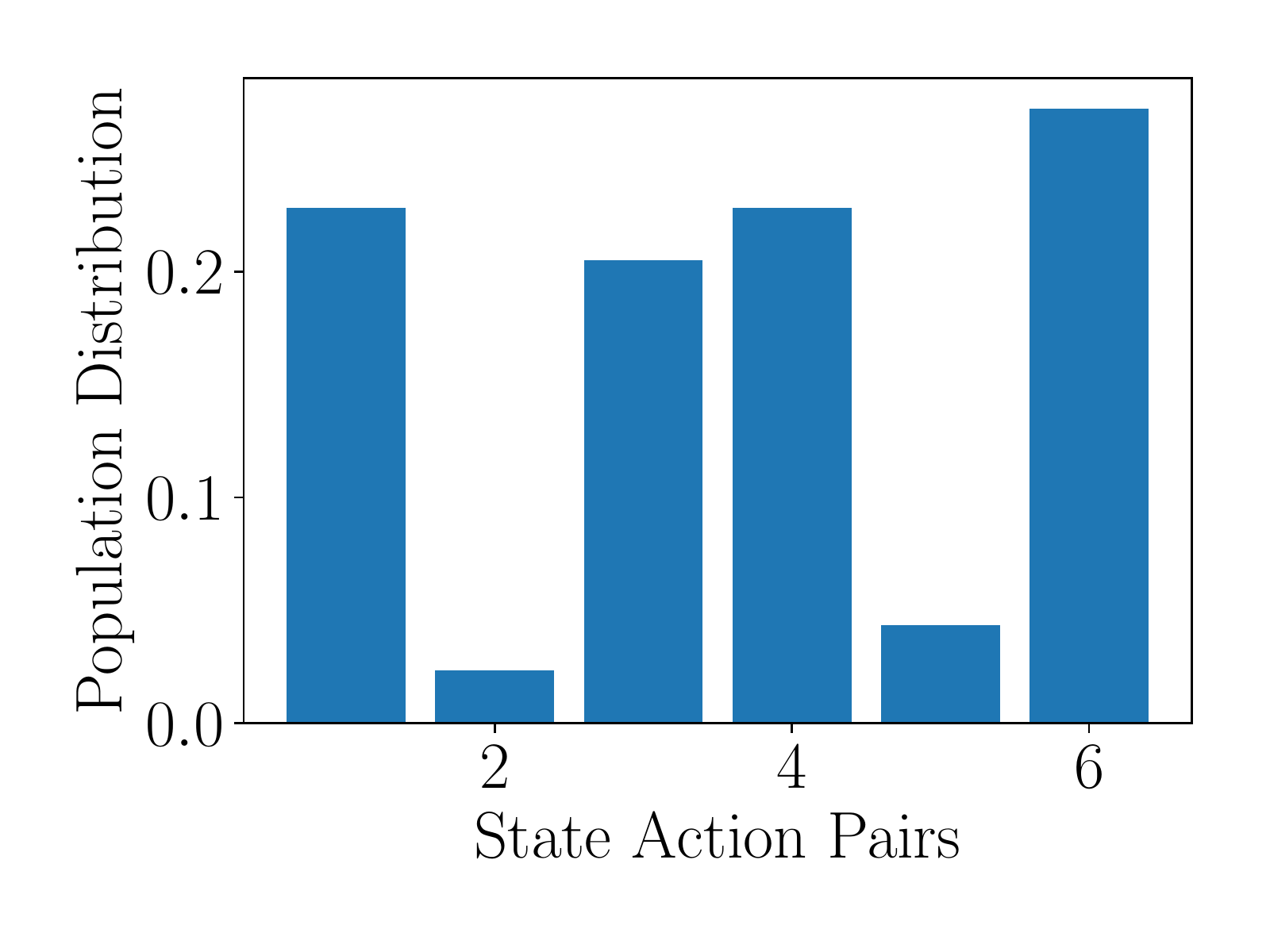}
    \caption{Optimal population distribution at with link costs from table~\ref{tab:linkCosts}}
    \label{fig:mdpPositiveOptimalDistribution}
\end{figure}

We consider perturbing the hyperarc costs modelled by $\bar{\ell}(\cdot,\epsilon) = \ell(\cdot) + \epsilon$.
Sensitivity of social cost can be analytically derived from Theorem~\ref{thm:sensitivity} based on the hypergraph structure as  
$\nabla_{\epsilon}J = \begin{pmatrix}0.023 & 0.501& -0.478 & 0.023& 0.454& 0.477\end{pmatrix}^T$.
\begin{figure}[ht]
    \centering
    \includegraphics[trim=0 35 0 30, width=0.8\columnwidth]{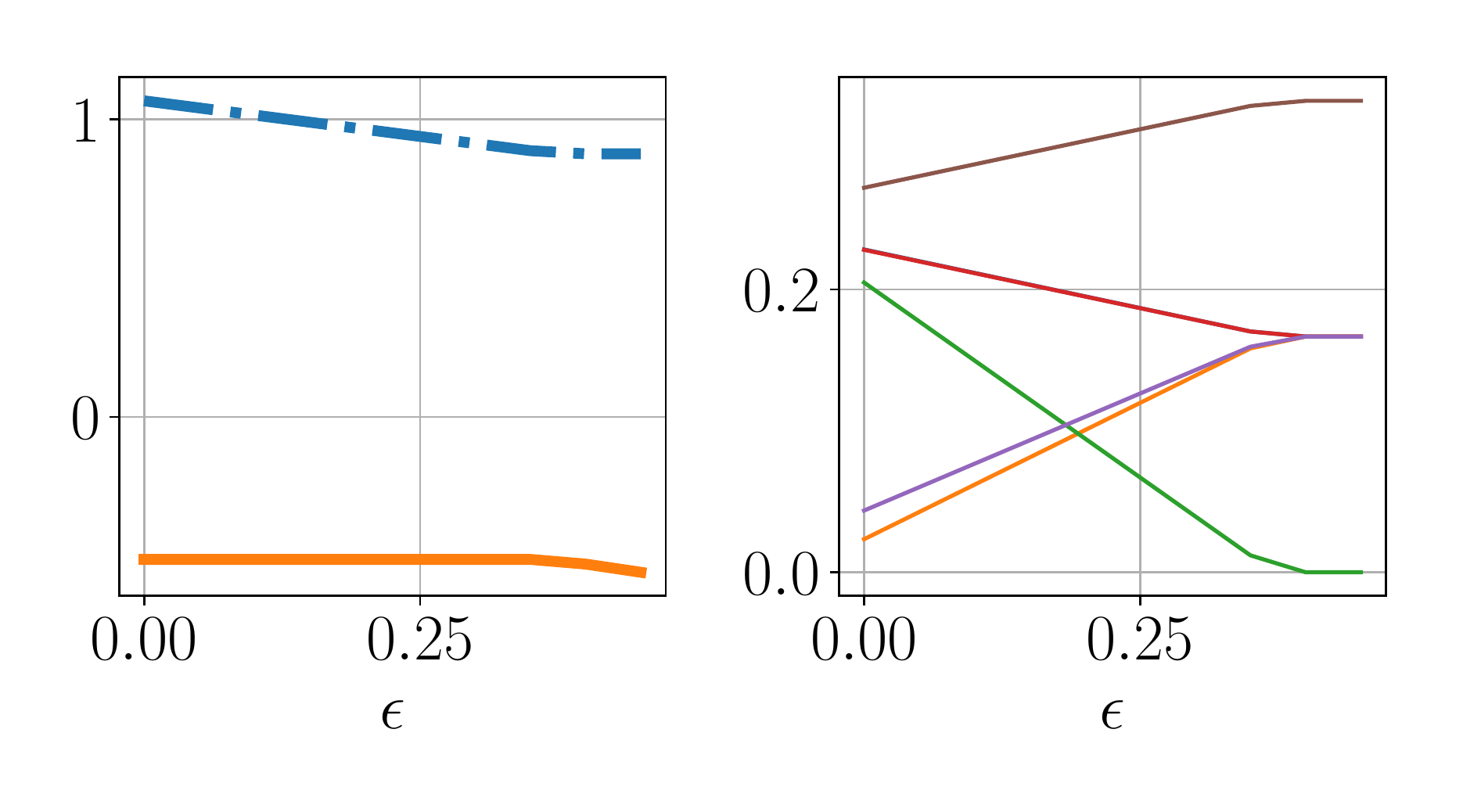}
    \caption{Braess Paradox: Perturbing game costs with $\epsilon[0,\, 0,\, 1,\, 0,\,0,\,0]$, where $\epsilon \in \reals_+$ increases along x-axis. Right shows the game optimal population distribution on each hyperarc. Left shows the social cost at optimal population distribution (blue) and the sensitivity for hyperarc $3$ varying with $\epsilon$ (orange).}
    \label{fig:braessParadox}
\end{figure}
The sensitivity vector $\nabla_{\epsilon}J$ implies that increasing the third hyperarc cost would result in the most decrease in social cost, while increasing the second hyperarc cost would result in the most increasing in social cost. We verify both scenarios by successively increasing $\epsilon$ and re-evaluating the social cost at the optimal population distribution $y^\star(\epsilon)$, as solved by cvxpy. The results are shown in Figures~\ref{fig:braessParadox} and \ref{fig:noParadox}.
\begin{figure}[ht]
    \centering
    \includegraphics[trim=0 35 0 20, width=0.83\columnwidth]{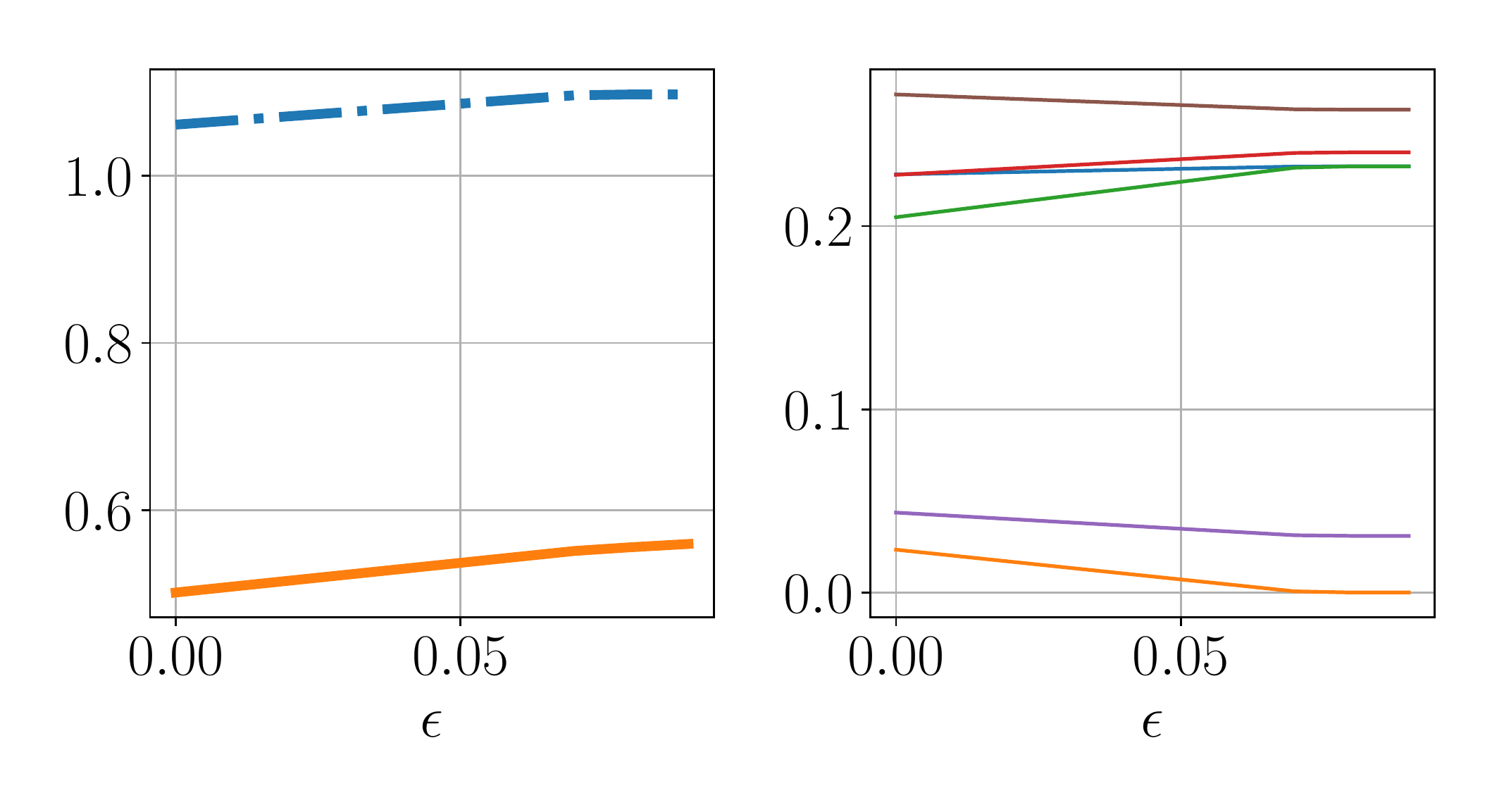}
    \caption{No Braess Paradox: Perturbing the game costs with $\epsilon[0,\, 1,\, 0,\, 0,\,0,\,0]$, where $\epsilon \in \reals_+$ increases along x-axis. Right shows the game optimal population distribution on each hyperarc. Left shows the social cost at optimal population distribution (blue) and the sensitivity value for hyperarc $2$ at given $\epsilon$ (orange).}
    \label{fig:noParadox}
\end{figure}

A couple conclusions can be drawn from Figures~\ref{fig:braessParadox} and \ref{fig:noParadox}. First, we see that there exists a continuous region around $\epsilon$ where $y^\star(\epsilon) >0$, and therefore renders this sensitivity analysis valid. Figure~\ref{fig:braessParadox} shows a negative sensitivity value for the third hyperarc as we increase $\epsilon$, which implies stochastic Braess paradox. Then as predicted, the social cost \emph{decreases} as $\epsilon$ is increased. In contrast, Figure~\ref{fig:noParadox} shows a positive sensitivity value for the second hyperarc as we increase $\epsilon$, therefore the social cost should not decrease as $\epsilon$ increases. This is also confirmed as the social cost obtained from the output of cvxpy increases with $\epsilon$. Both Braess paradox and the absence of Braess paradox is correctly predicted for the regions where positive mass exists on every hyperarc. 

\section{Conclusions}\label{conclusion}
We derived sensitivity analysis for MDP congestion games when the optimal population distribution is strictly positive. From the sensitivity of optimal cost and population distribution to changes in state-action cost, we derived sufficient conditions for the occurrence of stochastic Braess paradox defined in terms of network and cost structure. Finally, we considered effects of stochasticity on the magnitude of Braess paradox. Our simulations explicitly show the occurrence of stochastic Braess paradox on MDP congestion games. Future work include generalizing the analysis to MDP congestion games whose optimal population distribution is not strictly positive.

\bibliographystyle{IEEEtran}
\bibliography{reference}

\end{document}